\documentclass[12pt,twoside,letter]{amsart}
\usepackage{txfonts}
\usepackage{comment}
\usepackage{bbm}
\usepackage{mathrsfs}
\usepackage{amsfonts}
\usepackage{amsmath}
\usepackage{amssymb}
\usepackage{wasysym}
\usepackage{psfrag}
\usepackage{graphics,epsfig,amsmath}  
\usepackage{comment}
\usepackage{enumerate}

\newtheorem{theorem}{Theorem}
\newtheorem{definition}[theorem]{Definition}

\newtheorem{proposition}[theorem]{Proposition}
\newtheorem{preremark}[theorem]{Remark}
\newenvironment{remark}{\begin{preremark}\rm}{\end{preremark}}

\title[Ground states in ferromagnetic media]
{A continuous family of equilibria in ferromagnetic media
are  ground states}

\author[X. Su] {Xifeng Su}
\thanks{X. S supported by National Natural Science Foundation of China (Grant No. 11301513) and ``the Fundamental Research Funds for the Central Universities"}
\address{School of Mathematical Sciences\\
Beijing Normal University\\
No. 19, XinJieKouWai St.,HaiDian District\\
 Beijing 100875, P. R. China}
\email{xfsu@bnu.edu.cn}

\author[R. de la Llave]{Rafael de la Llave}
\thanks{R.L. has been partially supported by 
NSF grant DMS-1500943. R.L. also akcnowledges the support of 
the Tang Aoqing visiting professorship in Jilin University.}
\address{School of Mathematics\\
Georgia Institute of Technology \\
686 Cherry St. \\
Atlanta GA 30332, USA}
\address{
JLU-GT joint institute for Theoretical Science\\
Jilin University\\
Changchun, 130012, CHINA}
\email{rafael.delallave@math.gatech.edu}

\begin{document}

\today

\maketitle

\begin{abstract}
We show that a foliation of 
equilibria  (a continuous family of equilibria 
whose graph covers all the configuration space) 
in ferromagnetic models are ground states. 

The result we prove is very general, and it applies 
to models with long range interactions and many body. 
As an application,
we consider several models of networks of interacting particles including models of Frenkel-Kontorova type on $\mathbb{Z}^d$ and one-dimensional quasi-periodic media.

The result above  is an analogue of several results in the calculus
variations (fields of extremals) and in PDE's. Since the models 
we consider are discrete and long range, new proofs need to 
be given.
We also note that the main hypothesis of our result
(the existence of foliations of equilibria) is the
conclusion (using KAM theory) of several recent papers.  Hence, we 
obtain that the KAM solutions recently established are minimizers 
when the interaction is ferromagnetic (and transitive). 
\end{abstract}



\section{Introduction}
Many physical problems lead to variational problems for functions described 
in discrete sets. 

A model to keep in mind  as motivation is 
the Frenkel-Kontorova model \cite{FK} 
which considers  configurations $u=\{u_i\}_{i\in \mathbb{Z}} $
and tries  to find those  that 
minimize the energy given by the formal sum
\begin{equation}\label{FK} 
\mathscr{S} (u) = \sum_{i\in \mathbb{Z}}  \left[\frac{1}{2}(u_i-u_{i+1})^2 - V(u_i)\right]. 
\end{equation} 

There are several physical interpretations of the FK model
 \cite{BraunK,Selke2}, the
original one is the interaction of planar dislocations in a 3-D
crystal, but it has appeared as a model of other situations.  We can
think of $u_i$ as describing the position of the $i^{\rm{th}}$ atom
deposited over a $1-D$ medium.  The first part of the sum describes
the interaction between the nearest  particles. The function
$V$ models the interaction of the atoms with the medium, which is assumed to 
be periodic or quasi-periodic in models of crystals and 
quasi-crystals (in this paper, periodicity or 
quasiperiodicity is not assumed). 
 Note that, with many of these interpretations, 
it is natural to consider also more general models which involve 
longer range interactions, multi site interactions, or higher dimensional 
crystals.  Hence, in this paper we will include these generalizations.

In the case that $V$ is a periodic function, the problem of showing
existence of plane-like minimizers (i.e. minimizers that differ from a linear 
function by a bounded function) of \eqref{FK} with a well defined frequency
independently by Mather \cite{Mather'82} and Aubry \cite{ALD}, which
is now referred as Aubry-Mather theory. Several authors ( see
\cite{Blank'89, Blank'90,KLR'97, CLlave'98,LV'07,LV'07b,LV'10} and
references therein) generalized the setting of Aubry-Mather theory to
higher dimensional crystals, more general media and for many-body
interactions. Related models appear in PDE's \cite{Moser89, 
RabinowitzS}, minimal surfaces \cite{CaffarelliL, Valdinoci, Torres}, 
fractional laplacian operators \cite{ValdinociL, Davila}. 

In the case that $V$ in \eqref{FK} is a quasi-periodic function, the
problem to establish all the results of Aubry-Mather theory for
periodic systems is still open.  Notably, the existence of plane-like
minimizers is still not settled.  In \cite{LS'03} there are examples of
quasi-periodic potentials for which no plane-like minimizer exists. On
the other hand, when the potential $V$ is small enough, the papers
\cite{SuL1, SuL2, SuZL15,ZhangSL15} use a rather unusual KAM theory to
construct families of equilibria which are plane-like. The results of
this paper show that the families constructed by KAM method are
minimizers when the problem is ferromagnetic.  Hence a very
interesting problem is to study the transition -- now known to exist
-- from models with plane-like minimizers to models without them. The
papers \cite{SuL1, SuL2} lead to efficient numerical algorithms which
were implemented in \cite{BL'13} and lead to several conjectures about
the transition between plane-like and non plane-like
minimizers. Notably \cite{BL'13} discovered numerically scaling
relations similar to those in phase transitions in the breakdown of
analyticity of plane-like solutions in quasi-periodic media. The papers 
\cite{SuZL15,ZhangSL15} also present efficient algorithms for the computation of other 
solutions, but they have not been implemented yet. 

The goal of this paper is to show that for ferromagnetic models when
there are continuous families of equilibria  whose graphs  cover the whole phase
space they are actually ground states (also called class-A
minimizers). In particular, the solutions produced by KAM theory
in \cite{SuL1,SuL2,SuZL15,ZhangSL15} are
ground states.

Results establishing that families 
of equilibria are minimizers  are very common place in the standard calculus of
variations. They are proved by either the methods of \emph{fields of
  extremals} or \emph{Hilbert integral} \cite{Caratheodory}. In our
case, since we are considering discrete space and long range models,
these methods do not seem to apply directly and we have to use a
different method. 

We note that, as it is customary,  the
non-degeneracy equations of KAM theory are weaker than those in the
variational theory. Roughly speaking, the KAM theory just requires
that certain operators are invertible.  The variational theory
requires that these quantities are positve definite. 
On the other hand, the KAM theory is more sensitive to quantitative 
features. For example, in \eqref{FK} and periodic $V$, 
the KAM only applies for 
$V$ which are small in a smooth norm, 
whereas the variational methods apply for any differentiable $V$. 

In Section~\ref{general}, we present the results in a very general set
up, patterned after the general set up of statistical mechanics
\cite{Ruelle} allowing multi-body and long range interactions. 
In
Section~\ref{sec:examples} we present again the results for some
concrete models, which have appeard in the literature.  Even if this
could have been avoided logically since the models in
Section~\ref{sec:examples} are particular cases of those in
Section~\ref{general} we hope that this will add to the readability of
the paper and as motivation for those interested in the concrete
models and in numerical implementations. Also, the methods of proof
used in  Section~\ref{sec:examples} are different from the methods
 used in the proof of the general theorem and closer to the arguments
 in the classical calculus of variations.

\section{Formulation of the main result} 

\subsection{A very general set-up}\label{general}
We consider  a very general setup motivated by the
formulation in  \cite{Ruelle} of 
statistical mechanics.  Later, in Section~\ref{sec:examples}, 
we will present more details for 
less general set ups, which may be more familiar.

\subsubsection{General assumptions on the systems and its configurations}
We consider a discrete countable set $\Lambda$. Its elements will 
be called sites. The set $\Lambda$ may be imagined as a network  of 
particles. Many cases in statistical mechanics consider that 
$\Lambda$ is an integer lattice, corresponding physically 
to a crystal. 

We assume that the state of each site is given by a real
number. Hence, the state of the system is given by a function $u:
\Lambda\rightarrow \mathbb{R}$ which assigns to each site
$i\in\Lambda$ the value $u_i$. For our purposes, it is crucial that
the order parameter at each site is a one-dimensional number.  We do
not know how to deal with two-dimensional phase spaces.  Indeed in
\cite{Blank'90} presents counterexamples to several crucial statement
in our setting when the order parameter are 2-dimensional. The papers
\cite{Mather'91, Mane'97} contains rather satisfactory analogues of
several other results of Aubry-Mather theory for 
higher dimensional order parameters but they do not consider 
higher dimensional independent variables.

We associate to the finite subsets $B$ of $\Lambda$ an energy function
$H_B: \mathbb{R}^B \rightarrow \mathbb{R}$, 
which models the (possibly many-body and long range) interaction. 
In Physical terms, the interaction may be even among the different 
sites or among the sites and a substratum. 
The total energy associated to a configuration $u$ is given by the formal sum: 

\begin{equation}\label{formal energy}
\mathscr{S} (u) = \sum_{B\subseteq \Lambda\atop
\#B < \infty} H_B(u) \qquad \forall~ u\in\mathbb{R}^\Lambda
\end{equation}where $H_B(u)$ depends only on $u|_B$.

\begin{remark} In this paper, we will not assume any periodicity properties 
of the set $\Lambda$ and of the interaction, since this will 
not play any role in our arguments. On the other hand, we note that the 
main hypothesis of this paper  (the existence of a
foliation of equilibria)  is the conclusion of several 
other papers which use periodicity. 
In~\cite{CLlave'98,LV'10}, there is a very general setup for
quasi-periodicity involving the action of a group $G$ on 
$\Lambda$ and on the interaction.  $G$ is assumed to satisfy some mild growth properties.

In  Section~\ref{sec:examples} we will present the results for 
some finite range models which are concrete examples of 
the set up and for which our main hypotheses are
verified. 
\end{remark} 

\subsubsection{Critical points and ground states} 
The following definitions are very standard in the calculus of 
variations. 

\begin{definition} 
We say that a configuration $u$ is an equilibrium for 
an energy \eqref{formal energy} when 
\begin{equation}\label{equilibrium equation}
\frac{\partial}{\partial u_i} \mathscr{S}(u) \equiv \sum_{B\ni i} \partial_{u_i} H_B(u) =0 \qquad \forall~i\in\Lambda.
\end{equation}
\end{definition} 

For simplicity, we denote $\frac{\partial}{\partial u_i}
\mathscr{S}(u) $ by $E_i(u)$ and $E(u)=\{E_i(u)\}_{i\in \Lambda}$.

We note that, even if the sum defining $\mathscr{S}$ is formal, the
equilibrium equations \eqref{equilibrium equation} are meant to be well
defined equations.  This can happen for example if $H_B \equiv 0$
whenever $\text{diam}(B)\geq L$. (These are called finite range
interactions and Frenkel-Kontorova models are an example.) In
Section~\ref{sec:regularity}, we will formulate a condition, more
general than finite range which is enough for our purposes.

We are interested in the existence of the following special class of equilibria.
\begin{definition}[Ground states]
We say that a configuration $u$ is a ground state ( or a class-A minimizer in the terminology of Morse \cite{Morse'24}) if 
for any configuration $\varphi$  whose support is a finite subset of $\Lambda$ we have
\begin{equation}\label{ground-state-def}
\mathscr{S}(u) - \mathscr{S}(u+\varphi) \leq 0  .
\end{equation}
\end{definition}

Note that \eqref{ground-state-def} should be understood cancelling 
all the terms that are identical. 
That is 
\begin{equation}\label{ground-state-precise} 
\sum_{\#B < \infty \atop
B \cap \text{supp}(\varphi) \ne \emptyset} H_B(u) -  H_B( u + \varphi)  \le 0.
\end{equation} 
We note that the conditions~\eqref{ground-state-precise} make sense 
when the interactions  are finite range since the sum in \eqref{ground-state-precise} involves only finitely many terms. 
In Section~\ref{sec:regularity}, 
we will make assumptions  more general than 
finite range that ensure that the 
sum in \eqref{ground-state-precise} make sense.  It is clear 
that  the main idea is that we will assume the terms in the sum 
\eqref{ground-state-precise}  as well as their derivatives 
decay fast enough for all $u$ in a class of functions. We will 
postpone the precise formulation till we have specified which classes
of functions we will consider. 

Since expressions similar to \eqref{ground-state-precise}  
will appear often in our calculations, we will introduce the 
notation
\begin{equation}\label{gammanotation} 
\Gamma(\varphi; u, \tilde B)  \equiv
\sum_{\#B < \infty \atop
 B\cap \tilde B  \ne \emptyset} H_B(u) -  H_B( u + \varphi). 
\end{equation} 

We remark that if $\text{supp}(\varphi)  \subset \tilde B$, we have 
\[
\Gamma(\varphi; u, \tilde B)   =  
\Gamma(\varphi; u, \text{supp}(\varphi) )  .
\]
The reason is that, the sums defining the two $\Gamma$ differ only in 
sets $B$ which do not intersect the support of $\varphi$.  Hence, the corresponding term in the sum is 
zero. 

\medskip

It is easy to check that a ground state is  an equilibrium. 

\subsubsection{Foliations by equilibria} 

We say that a 
collection of configurations $\{u^\beta\}_{\beta\in\mathbb{R}}$ 
is a foliation
when:

\begin{itemize}
\item [(A1)] $E(u^\beta) =0$, i.e. $E_i(u^\beta) =0$ for any $\beta\in \mathbb{R}, i\in\Lambda$;

\item [(A2)] $u^\beta$ is increasing with respect to $\beta$, i.e. if $\beta_1\leq\beta_2$, $u_i^{\beta_1} \leq u_i^{\beta_2}$ for any $i\in\Lambda$;

\item [(A3)] $u_i^\beta \rightarrow \pm \infty$ as $\beta$ goes to $\pm \infty$ for any $i\in \Lambda$;

\item [(A4)] ${\left\{ u_i^\beta ~\big|~ \beta\in\mathbb{R} \right\} } = \mathbb{R}$.
\end{itemize}

The most crucial assumption for us is (A4).  This means that, as we
move the parameters $\beta$, the graphs of the functions $u^\beta$
sweep out all the space $\Lambda \times \mathbb{R}$. 

We say a foliation is strict if 
\begin{itemize}
\item [(A2)']  $\beta_1< \beta_2 \Longrightarrow u_i^{\beta_1} < u_i^{\beta_2}$ for all $i\in\Lambda$.
\end{itemize}

Having a family of critical points satisfying (A1)-(A4) is extremely
analogous to the assumption on the \emph{fields of extremals} in
the calculus of variations
\cite{Caratheodory}\footnote{Note that the fields of extremals in \cite{Caratheodory} 
are formulated for functions of a one dimensional variable taking values 
into any dimensional space. Here we are in the opposite situation: 
we are considering functions of many variables, but taking values in a one
dimensional space.}.

The usual Aubry-Mather theory for a fixed frequency $\omega$ produces a family 
satisfying (A1)-(A3) -- but in general not (A4). On the other hand, for 
Diophantine $\omega$ and (and some models) we can use KAM theory to produce families satisfying 
(A1)-(A4). The calculus of variations methods do not assume 
that the system is close to integrable, but they require positive 
definite assumptions on the interaction. On the other hand, KAM 
methods do  not require that the system is convex (positive 
definite Jacobian) but they require that the system admits an approximate solution to the invariance equation (in particular, this is satisfied for systems close to integrable).

In the applications to Aubry-Mather theory which we will discuss 
later in Section~\ref{sec:examples}, the set $\Lambda$ will be 
$\mathbb{Z}^d$ and the functions $u^\beta$ will be roughly linear. 

We point out, however that in the case of no interactions, in dimensions 
bigger or equal than $2$ one could have also harmonic polynomials,
which are minimizers. It is marginally pointed out in
\cite{Moser86,Moser89} that developing a variational theory
starting from the harmonic polynomials of higher 
degree would be very interesting. 

Note that the subsequent properties we will assume depend on the 
class of functions $u^\beta$.

\subsubsection{Ferromagnetic properties} 
\label{sec:ferromagnetic} 

\begin{definition}[Ferromagnetic condition]
We say that the $C^2$ interaction potential $H$ satisfies the ferromagnetic condition if 
\begin{equation}
\frac{\partial^2 H_B}{ \partial u_p \partial u_q}(u)  \leq 0  \qquad \forall p,q\in \Lambda, p\neq q,
\end{equation}where $u$ is any configuration on $\Lambda$ and $B$ is any finite subset of $\Lambda$.
\end{definition}

\begin{definition}[Ferromagnetic Transitive]\label{Ftransitivity}
We say that a ferromagnetic  interaction in $\Lambda$ is transitive 
for a class of configurations $u^\beta$ when, given any $p, q \in \Lambda$
there exist an integer $k\geq1$, a sequence $p_0, \ldots, p_k$ in $\Lambda$ with $p_0=p, p_k=q$ and sets $B_i$ containing a pair $p_i,p_{i+1}$ for $i=0,\ldots, k-1$ such that, for any $\varphi$ with compact support,
\[
\partial_{p_i} \partial_{p_{i+1}} H_{B_i} (u^\beta + \varphi)< 0.
\]
\end{definition}

In the main cases of interest, such as the Frenkel-Kontorova models, 
we will see that the 
$\partial_{p_i} \partial_{p_{i+1}} H_{B_i}$ are independent of the configuration,
so that this assumption will  be very easy to verify in several models of practical 
importance. 

The assumption in Definition~\ref{Ftransitivity} appeared in \cite{LV'07b} where it was shown that it implies 
that the gradient flow of the formal energy $\mathscr{S}$ 
in \eqref{formal energy} satisfies a strong comparison principle. 
In the PDE case, a comparison principle for the gradient flow 
would give a very quick proof of our results, but the long range of 
the interactions requires an extra argument. See Remark~\ref{partialrange}. 

\begin{remark} \label{rem:twist} 
The ferromagnetism assumptions, when $\Lambda = \mathbb{Z}$, 
and the interactions are nearest neighbor, become the twist conditions in 
Aubry-Mather theory. We also note that they can be thought of as analogues of 
ellipticity conditions for continuous variational problems. See \cite{CLlave'98, LV'07} 
for some more explanations of these  analogies. 
\end{remark} 

\subsubsection{Graph theoretic language to describe the Ferromagnetic assumptions} 
\label{sec:graph} 

We can reformulate some of the assumptions of
Section~\ref{sec:ferromagnetic} in the language of graph theory.  The
introduction of a new language is purely cosmetic, but allows us to
express future arguments concisely and it may be illuminating.

The key observation is that we 
can  interpret Definition~\ref{Ftransitivity} as  the existence of a graph
structure on $\Lambda$. 

Whenever 
there exists $B$ such that for all $u^\beta$
\begin{equation}\label{linked} 
\partial_{p}\partial_{q} H_B(u^\beta + \varphi) <0
\end{equation}
then the sites $p, q$ are linked.

The physical meaning of \eqref{linked} is that the configuration at
$p$ affects the forces experienced at the site $q$ (and viceversa, in
agreement with the action-reaction principle).  The
Definition~\ref{Ftransitivity} can be interpreted as saying that any
site can influence any other site, if not directly, through
influencing intermediate sites that in turn influence some others.

It is natural to endow $\Lambda$ with a graph structure by considering 
the points of $\Lambda$  as vertices 
and drawing an edge among two linked sites in the sense of
\eqref{linked}.

 The assumption in 
Definition~\ref{Ftransitivity} can be interpreted as saying that, starting from any site, we can reach any other 
jumping only through linked sites or that the graph is connected.

The graph structure allows us to introduce  two notions 
that are standard in graph theory: distance and connectedness. 

Given a path $\gamma$ in the graph, we define the $|\gamma|$ the
length of a path $\gamma$ as the number of edges it contains.

We define the \emph{distance between two sites} $i,j \in \Lambda$ 
as 
\begin{equation}\label{distance} 
d(i,j) = \inf\{~ |\gamma| \, ~\big|~  \gamma ~\text{joins}\ i, j\}.
\end{equation} 
This $d$ satisfies the usual assumptions of distance. 

We also define the \emph{distance of a  point $i$  to a set}  $S \subset \Lambda$
as 
\begin{equation}\label{distanceset} 
d(i,S) = \inf\{~ d(i, j) ~\big|~  j \in S \}.
\end{equation} 
(Since the $d(i, j)$ takes values in integers, it is clear that the infimum in \eqref{distanceset} is a minimum.)

We can also define that a set $S$  is \emph{connected} when any pair of points 
 can be joined 
by paths all whose edges have end points in the set $S$.

It will be important for us later 
that, given a finite set $B$,  we can obtain another finite set
$\text{Con}(B)$ which is connected and which contains $B$. 

If Definitions \ref{Ftransitivity} holds, given any pair $i,j\in
\Lambda$ we can find a path joining $i$ to $j$.  We denote this path
as $\gamma_{i,j}$. Given a path $\gamma_{i,j}$ we denote
$v(\gamma_{i,j})$ the \emph{vertices} of the path. Hence, given a set
$B$, we define
\begin{equation}\label{connected-hull} 
\text{Con}(B) = \cup_{i, j\in B} v(\gamma_{i,j} ).
\end{equation} 

Clearly, $B \subset \text{Con}(B)$ and $\text{Con}(B)$ is connected
because we note that given any pair of points $a \in \gamma_{i,j}, b
\in \gamma_{\tilde i, \tilde j}$ we can find a path joining them by
starting in $a$, following $\gamma_{i,j}$ till $j$, then $\gamma_{j,
  \tilde i}$ and then $\gamma_{\tilde i, \tilde j}$ till we arrive to
$b$.

The following elementary remark will play an important role for us later, so 
we formulate it now. It is mainly an exercise in the notation. 
\begin{proposition}\label{silly} 
Assume that the interaction satisfies Definition~\ref{Ftransitivity}. 

Given any finite set $S$, the set 
\[
S_1 = \{~i ~\big|~ d(i, S)  \le   1\} = \{~ i ~ \big| ~ j \in S, d(i, j) \le 1\}
\]
contains at least a point which is not in $S$. 
\end{proposition} 

The  totally trivial proof of Proposition~\ref{silly} is the observation that, if there 
was no path that stepped out of $S$, it would be impossible for any point in $S$ 
to be joined to other points far away. 
\qed

\subsubsection{Coerciveness Assumption} 
\label{sec:coerciveness} 

Given a family $u^\beta$ as before, we will assume 
that for any compactly supported $\varphi$
and any $i \in \text{supp}(\varphi)$ we have 
\begin{equation} 
\label{eq:coerciveness} 
\lim_ {|t| \to \infty} \sum_{B \cap \text{supp}(\varphi) \ne \emptyset }
\left[H_B(u^\beta + \varphi + \delta_i t ) - H_B( u^\beta + \varphi) \right]  = + \infty ,
\end{equation} where $\delta_i$ denotes the Kronecker function which takes the value 1 at $i$ and 0 at any other point.

Note that \eqref{eq:coerciveness} says that if we make  a test function grow at just 
one point, then the relative energy grows.

\subsubsection{A regularity assumption} 
\label{sec:regularity}

We will be performing some calculations with the equilibrium 
equations. In order to justify them, we will need some assumptions on the 
convergence of the $E_i$ and their derivatives. 

The following assumption is sufficient for the methods used in 
this paper. We note that the finite range of the interaction 
easily implies our assumption. Many 
models of interest (e.g.  the Frenkel-Kontorova models) are finite 
range, but there are aslo models of physical interest which are not. 
See \cite{SuL1} for  a discussion of when hyerarchical models satisfy the 
assumptions. 

Given a class $u^\beta$ satisfying (A1)-(A4) we say that the interaction 
$\{ H_B\}$ is $u^\beta$ summable when: for $\varphi$ satisfying either
\begin{itemize}
\item [a)] $\varphi$ with compact support 
\item [b)] $\varphi =  (u^{\tilde \beta} - u^{\hat \beta})$ for any $\tilde \beta, \hat \beta \in \mathbb{R}$,
\end{itemize}
we have for all $\beta \in \mathbb{R}$,
\begin{equation} 
\label{fast-decay}
\begin{split}
&\lim_{L \to \infty}
\sum_{\text{diam}(B) \ge L 
\atop B \cap \text{supp}(\varphi) \ne \emptyset }
|\partial_{u_i}  H_B( u^\beta + t \varphi )| = 0\\
&\lim_{L \to \infty}
\sum_{\text{diam}(B) \ge L 
\atop B \cap \text{supp}(\varphi) \ne \emptyset }
|\partial_{u_i} \partial_{u_j} H_{B}( u^\beta + t \varphi )| = 0\\
\end{split}
\end{equation} 
and the limit in \eqref{fast-decay} is uniform in $t \in [0,1]$ 

We note that in the case $u^{\tilde \beta} - u^{\hat \beta} \in \ell^\infty$
 (as it happens in KAM theory) and in all families of plane-like equilibria of fixed slope b) is implied by a). We also note that if we assume (A2)' instead of (A2), the case b) can be dispensed with.

The following result is a very simple corollary of the 
coercivity and regularity assumption.

\begin{proposition}\label{trivial} 
Let $u^\beta$ be a family of configurations and $\{ H_B\}_{B \subset \Lambda \atop \#B < \infty}$
be a family of
interactions that satisfy the coerciveness and the regularity assumptions
with respect to them. 

Fix any function $u^\beta$ in the family and a finite set $\tilde B$. 

Then, there is a function $\varphi^*$ such that 
\begin{itemize} 
\item 
$\text{supp}(\varphi^*) \subset \tilde B $ 
\item 
\[
\Gamma( \varphi^*; u^\beta, \tilde B) = \inf\{~ \Gamma( \varphi; u^\beta,\tilde  B) ~\big |~ 
\text{supp}(\varphi) \subset \tilde  B ~\}
\]
where we use the notation introduced in \eqref{gammanotation}.
\item
\begin{equation}\label{equilibriumB}
E_i( u^\beta + \varphi^*)  = 0 \quad \forall~i \in \tilde B .
\end{equation} 
\end{itemize} 
\end{proposition} 

The proof of Proposition~\ref{trivial} is very easy. We note that we
are considering a function of finitely many real variables (the values
of $\varphi^*$ at the sites of $\tilde B$). By the assumption of
regularity this function is differentiable and tends to infinity as
any of its arguments goes to infinity. Hence, this function reaches
its minimum and the minimum has zero derivative.

Of course, the support of the minimizing function could be smaller
than $\tilde B$ is some of the values of the miniming function happens
to be zero.

\begin{remark} \label{partialrange}
Note that in Proposition~\ref{trivial} we do not obtain that $u+ \varphi^*$ 
is an equilibrium. In \eqref{equilibriumB}, we only obtain that the 
equilibrium equations hold in the finite set $\tilde B$. 

Even if $u^\beta$ satisfies the equilibrium equations in 
$\Lambda$,  when the interaction is 
long range, modifying the configuration in $\tilde B$ can affect the 
equilibrium equations everywhere. 

This is an important difference with the PDE models in the classical calculus
of variations and this a the reason why our arguments need to be different. 
\end{remark}

\subsubsection{Statement of the main general result}

\begin{theorem}\label{main}
Let $H$ be a $C^2$ ferromagnetic interaction potential.
Assume that there exists a collection of configurations $\{u^\beta\}_{\beta\in\mathbb{R}}$ such that  (A1)-(A4) hold. 
Moreover, assume that, with respect to $u^\beta$ the interaction satisfies the ferromagnetic transitivity, coercivity and regularity assumptions above.

Then, all the equilibria $u^\beta$ are ground states.
\end{theorem}

Suppose by contradiction that there exist a number $\beta_0$ and a configuration $\varphi$ whose support is nonempty and finite such that 
\begin{equation}
\mathscr{S}(u^{\beta_0} + \varphi) - \mathscr{S}(u^{\beta_0}) < 0.
\end{equation}
That is, using the notation \eqref{gammanotation}
\begin{equation}
\Gamma(\varphi; u^{\beta_0}, \text{supp}(\varphi) ) < 0. 
\end{equation}

Denote by $\tilde B  = \text{Con}(\text{supp}(\varphi))$ the connected subset constructed in 
\eqref{connected-hull}, $\text{supp}(\varphi) \subset \tilde B$.

Using Proposition~\ref{trivial} 
there is a configuration  $\varphi^*$  with support in 
$\tilde B$ such that
\begin{equation}
\Gamma(\varphi^*; u^{\beta_0}, \tilde B) = \min_{\text{supp}(\varphi_1) \subseteq \tilde B} 
\Gamma(\varphi_1; u^{\beta_0}, \tilde B).
\end{equation}
We note that, since we can take $\varphi$ as a test function $\varphi_1$ we have
\begin{equation}\label{varphistarnottrivial}
\Gamma(\varphi^*; u^{\beta_0}, \tilde B) = \Gamma(\varphi; u^{\beta_0}, \tilde B)  = 
\Gamma(\varphi; u^{\beta_0}, \text{supp}(\varphi)) )  < 0 .
\end{equation}

Hence, if the function $u^{\beta_0}$ was not a ground state, we could find a non-trivial 
$\varphi^*$. We will show that this is impossible and, therefore that 
$u^{\beta_0}$ is a ground state.

We denote
\begin{equation}
\begin{split}
\beta_+ &= \inf\{ \beta\in \mathbb{R} ~|~ u^\beta > u^{\beta_0} + \varphi^* \},\\
\beta_- &= \sup\{ \beta\in \mathbb{R} ~|~ u^\beta < u^{\beta_0} + \varphi^* \},
\end{split}
\end{equation}where the partial ordering $u< v$ is defined by $u_i < v_i$ for any $i\in \Lambda$. Analogous definitions hold for ``$>$", ``$\geq$" and ``$\leq$".
Consequently, we have that assumption (A2) can be formulated just 
as  $u^{\beta_+} \geq u^{\beta_0} \geq u^{\beta_-}$.

By the choice of $\varphi^*$ and $\beta_+$, we have 
\begin{equation}
\begin{split}
E_i (u^{\beta_0} + \varphi^*) &=0, \quad i \in \tilde B \\
E(u^{\beta_+}) & = 0.
\end{split}
\end{equation}

Moreover, we have $u^{\beta_0} + \varphi^* \leq u^{\beta_+}$. 

The following is an elementary calculation using the 
fundamental theorem of calculus which holds for any configuration 
$u^*$ and any $\eta$ so that the regularity assumptions hold.

\begin{equation}\label{integralequation} 
\begin{split}
E_{i^*}& (u^{*} + \eta ) - E_{i^*}(u^{*})\\
&= \int_0^1 dt \left[ \sum_{ j\in \Lambda} 
 \frac{\partial^2 H_B}{\partial u_{i^*} \partial u_j} ( u^{*}+ t \eta ) \eta_j \right] \\
&=  \eta_{i^{*}} \int_0^1 dt  
\frac{\partial^2 H_B}{\partial u_{i^{*}} \partial u_{i^{*}}} ( u^{*}+ t \eta )  \\
&+ \sum_{j \in \Lambda \atop j \ne i^*} \eta_j 
\int_0^1 \frac{\partial^2 H_B}{\partial u_{i^*} \partial u_j} ( u^{*}+ t \eta ) .
\end{split}
\end{equation}

The identity \eqref{integralequation} leads immediately to 
the following proposition. 

\begin{proposition}\label{contact} 
Assume that, in the conditions of \eqref{integralequation} we have 
\[
\begin{split} 
& E_{i^*}(u^*) = E_{i^*}(u^* + \eta )  \\
& \eta \ge 0~ ( \text{ or } \eta \le 0~)  .
\end{split} 
\]
Then, we have that $\eta_j = 0$ for all $j$ such that $d(i^*,j) = 1$
where $d$ is the graph distance introduced in \eqref{distance}. 
\end{proposition}

The proof of Proposition~\ref{contact} is just observing that since 
$\eta_{i^*} = 0$, and   all the other terms in \eqref{integralequation} have 
the same sign, we should have that all of the terms in the sum in 
\eqref{integralequation} should be zero. Hence, either $\eta_j = 0$ or
$\int_0^1 \frac{\partial^2 H_B}{\partial u_{i^*} \partial u_j} ( u^{*}+ t \eta )$, 
but for the points $j$ at distance $1$, this integral is not zero. 
\qed

Applying repeatedly Proposition~\ref{contact} we have the following result for functions 
which satisfy the equilibrium equation on a set. 

\begin{proposition}\label{contact2} 
Assume that, in the conditions of \eqref{integralequation} we have 
\begin{itemize} 
\item
\[
\begin{split} 
& 0 =  E_{i^*}(u^*) = E_{i^*}(u^* + \eta ) \quad \forall~ i^* \in \tilde S  \\
& \eta \ge 0 ~( \text{ or } \eta \le 0~)  .
\end{split} 
\] 
\item 
The set $\tilde S$ is connected in the sense introduced in Section~\ref{sec:graph}.
\end{itemize}

Then, we have that $\eta_j = 0$ for all $j$ such that $d(i^*, \tilde S) \le 1$
where $d$ is the graph distance introduced in \eqref{distance}. 
\end{proposition} 

The proof of Proposition~\ref{contact2} is to proceed by induction
starting on the point $i^{*}$. Applying Proposition~\ref{contact} we
obtain that for all the points $j$ such that $d(i^*, j) = 1$, we
should have $\eta_j = 0$. Now, for such points $j$ that belong to $S$,
the argument can restart. Therefore, proceeding by induction, we can
always prolong the paths that land in $S$. Because $S$ is connected,
we can cover all the set $S$ and obtain that $\eta_j = 0$ for all the
$j$ in $S$. Note also that in the last step, we can get also that
$\eta$ vanishes in the set of points that are at distance $1$ from
$S$.  It will be important for future purposes that, as observed in
Proposition~\ref{trivial}, we have that the set where we can obtain
that $\eta = 0$ is strictly larger than the set $S$.

\qed

\begin{remark} 
With the analogies in Remark~\ref{rem:twist}, we note Proposition~\ref{contact2} 
is reminiscent to the proof of the comparison principle for elliptic equations. 
Of course, the proof in the discrete case is different. The subtlety that 
we obtain the comparison in a larger set than the set where the equation holds 
does not have any analogue in the elliptic equations case. 
\end{remark} 

Now, we come back to the proof of  Theorem~\ref{main}.

Since we have that $u^{\beta_+} \ge u^\beta_0 + \varphi^*$ 
and that the $\beta_+$ is the smallest possible,  we have alternatives: 
\begin{itemize} 
\item[A)]  There is a point where $\varphi^*$ is strictly  positive;
\item[B)]  $\varphi_i^* < 0$ for all $i \in \tilde S$;
\item[C)] $\varphi^* \equiv 0$. 
\end{itemize} 

Theorem~\ref{main} will be established when we show that all these
alternatives are impossible.  Hence, we conclude that $\varphi^*$ in
\eqref{equilibriumB} could not exist and, hence no $\varphi$ satisfying 
\eqref{contradiction} could exist.

The case C) can be excluded becasuse we argued in \eqref{varphistarnottrivial} that $\varphi^*$ should be 
non trivial. 

In case A), there exists $i^* \in \tilde S$ such that 
\[
u^{\beta_+}_{i^*} = u^{\beta_0}_{i^*} +
 \varphi^*_{i^*}.
\] 
In this case, recalling that $u^{\beta_+}$
satisfies  the  equilibrium equations in the whole 
$\Lambda$ and that $u^{\beta_0}+ \varphi^*$ satisfies 
them in $\tilde S$, we can apply Proposition~\ref{contact2} with $\eta= \varphi^*$ and obtain that 
\[
 u^{\beta_+}_{j} = u^{\beta_0}_{j} + \varphi^*_{j}, \quad d(j, \tilde S) \le 1.
\]

The important point of the above observation is that there is such a point 
$j$ outside of $\tilde S$. That is, a point $j$ outside of the support of 
$\varphi^*$. Hence, there is point $j^*$ such that 
\begin{equation}\label{identicalj}
u^{\beta_+}_j = u^{\beta_0}_j.
\end{equation}

When we apply Proposition~\ref{contact2} with $\eta= u^{\beta_+} - u^{\beta_0}$ we obtain that $u^{\beta_+} = u^{\beta_-}$. 
This is a contradiction with $\varphi^*$ being strictly positive. 

Note that if we assume (A2)' from \eqref{identicalj} we could obtain the conclusion without applying Proposition~\ref{trivial}.

Excluding Case B) is very similar to excluding case A), but actually easier. 
Since there is point $j$ where $\varphi^*_j$ is strictly negative, we can 
have that there is a $j^*$ where $u^{\beta_-}$ touches from below 
the $u^{\beta_0} + \varphi^*$. Applying again Proposition~\ref{contact2}, we derive
that $u^{\beta_-} = u^{\beta_0}$ which is a contradiction with the assumption that 
$\varphi^*$ was strictly negative.

\section{Some concrete examples of the models considered}\label{sec:examples}
In this section, we will show how very different models fit simultaneously in the framework developed.
The fact that we can obtain results for different models at the same times is due to the generality of the methods we present here. In some cases, we will also present different proofs.

\subsection{General one-dimensional periodic models}
The papers \cite{Rafael'08} considers one dimensional models
given by energies of the form:
\begin{equation}\label{formal functional}
\mathscr{L} (u) = \sum_{k} \sum_{L} H_{L, k} (u_k, u_{k+1}, \ldots , u_{k+L})   
\end{equation}where $u: \mathbb{Z} \rightarrow \mathbb{R}$ and $H_{L,k}$

Note that the models in \eqref{formal functional} enjoy a
 translation invariance, which is not present in our general set up, but which is physically justified.

The corresponding equilibrium equation for the models \eqref{formal functional}
are: 
\begin{equation}
\mathscr{E}_i (u) = \sum_{k} \sum_{L} \sum_{j} \partial_{j+1} H_{L, k} (u_{i-j}, \ldots, u_i, \ldots, u_{L-k+i-j}).
\end{equation}

The paper \cite{Rafael'08} includes coercivity and regularity similar
to ours,  it includes an extra periodicity property 
\[
H_L(u_k, u_{k+1}, \ldots u_{k+L}) = 
H_L(u_k + 1 , u_{k+1} + 1 , \ldots u_{k+L} + 1)
\]
as well as higher regularity assumption. On the other hand, the 
paper \cite{Rafael'08}  does not use the full strength of the 
ferromagnetic property and 
indeed they allows some antiferromagentic terms. 
Note that, since the systems in \cite{Rafael'08} are 
translation invariant, the ferromagnetic transitive is 
implied by the ferromagnetic property of nearest 
neighbors (there are other assumptions such as the 
strict ferromagnetic for other sets of interactions). 

The papers \cite{Rafael'08} consider only equilibriun configurations 
given by a hull function
\[
 u_k = \omega k + h(\omega k)
\]
where $h$ is a periodic function called the \emph{``hull funcion''}. 
The function is such that $t + u(t)$ is an increasing function. 

It is easy to see that -- it is shown with many details  in \cite{Rafael'08} 
that if  $h$ is the hull function for a critical point so is 
$u^\beta$ given by 
\[
h^\beta(\theta) = \beta  + h(\theta + \beta). 
\]
We observe that, when $h$ is a smooth function and $|h|_{L^\infty} <1$, the configurations 
obtained for all these hull functions produce a foliation in our sense. 

Hence, applying Theorem~\ref{main}, we obtain the following result: 

\begin{theorem} 
Assume the setup of \cite{Rafael'08} assume furthermore, that, for 
some hulll function, the 
system satisfies the 
ferromagnetic property
\[
\partial_i \partial_j  H_L   \le 0 
\]
and that
\[
\partial_1 \partial_2 H_1(x, y) \le -\eta < 0 
\]

Then, the quasi-periodic 
solutions produced in \cite{Rafael'08} are ground states. 
\end{theorem}

\subsection{Application to the Frenkel-Kontorova models on quasi-periodic media}

These class of models was considered in \cite{SuL1} with nearest neighbor interactions. In \cite{SuL2} for more general interactions, many body interactions.
The papers \cite{SuL1,SuL2} consider quasiperiodic solutions 
which are non-resonant (inded Diophantine) with the frequency of 
the medium. The papers \cite{SuZL15,ZhangSL15} study quasi-periodic 
solutions which are resonant with the frequency of the medium in  models 
in which the interactions are only nearest neighbor. 
Using the results of this paper, we can conclude that the solutions are 
ground states provided that we assume transitive ferromagnetic conditions. 

In this section, we will consider only the problem in \cite{SuL1}, 
which will allow us to give a more direct proof of the results. 
We note that in the models based on the Frenkel-Kontorova models 
with next neighbor interaction, the transitive ferromagnetic hypothesis is 
automatic. 

We consider the following formal energy
\begin{equation}\label{FK energy}
\mathscr{S}(\{u\}_{i\in
\mathbb{Z}})=\sum_{n\in\mathbb{Z}}\frac{1}{2}(u_n-u_{n+1})^2-V(u_n\alpha),
\end{equation}where $V: \mathbb{T}^d\rightarrow\mathbb{R}$ and $\alpha\in \mathbb{R}^d$ 
 satisfy $k\cdot \alpha\neq 0$ when $k\in\mathbb{Z}^d-\{0\}$ where $d\geq 2$.

For simplicity, we denote $H(x,y) = \frac{1}{2}(x-y)^2 - V(x\alpha)$. Consequently, $\partial_{xy} H(x,y)= \partial_{yx} H(x,y)= -1$.

Under the assumption of \cite{SuL1,ZhangSL15}, using KAM method, we prove the existence of quasi-periodic solutions 
of the equilibrium equation
\begin{equation}\label{equilibrium}
u_{n+1}+u_{n-1}-2u_n+\partial_\alpha V(u_n\alpha)=0,
\end{equation} where $\partial_\alpha V \equiv (\alpha\cdot
\nabla) V$.

Indeed, the solutions of \eqref{equilibrium} we found are given by a hull function
\[
u_n=n\omega + h(n\omega\alpha)
\]for some given $\omega\in \mathbb{R}$. Therefore, the equilibrium equation we solve in terms of $h$ is 
\begin{equation}\label{equilibriumhullforhat}
h(\sigma+\omega\alpha)+h(\sigma-\omega\alpha)-2h(\sigma)+\partial_\alpha
V(\sigma+\alpha\cdot h(\sigma))=0
\end{equation}

The papers \cite{SuL1, ZhangSL15} have very different non-resonance 
assumptions from the assumptions in \cite{SuL1, SuL2} and require very different methods. Nevertheless,
from the point of view of the arguments of this paper, to show 
that the quasi-periodic solutions produced in both papers
are ground states, we can use the same argument. 

It is easy to see that
if $h(\sigma)$ is a solution \eqref{equilibriumhullforhat}, for any
$\beta\in\mathbb{R}$, $h(\sigma+\beta\alpha)+\beta$
is a solution.  We denote $h_\beta (\sigma)= h(\sigma+\beta\alpha)+\beta$.

Hence, let us denote $u_n^\beta = n\omega + h_\beta(n\omega\alpha)$ which is a continuum of equilibria of \eqref{equilibrium} with respect to the parameter $\beta\in\mathbb{R}$. It is easy to see that, for every fixed $n\in\mathbb{Z}$, $u_n^\beta$ is monotone with respect to $\beta$, i.e.,
\begin{equation}
\frac{\partial u_n^\beta}{\partial \beta} = 1+ \partial_\alpha h(n\omega\alpha + \beta\alpha) \neq 0.
\end{equation}
Without loss of generality, we asume $u_n^\beta$ is monotone increasing with respect to $\beta$.

The following result is a particular case of 
Theorem~\ref{main}, but in this section, we will present a different proof. 

\begin{theorem}\label{QPFK}
For every $\beta\in \mathbb{R}$, the configurations $u^\beta\equiv\{u_i^\beta\}_{i\in\mathbb{Z}}$ are 
ground states of \eqref{FK energy}.
\end{theorem}
\begin{proof}
Suppose by contradiction that there exists $\beta_0$ such that $\{u_i^{\beta_0}\}_{i\in \mathbb{Z}}$ is not a ground state of \eqref{FK energy}. That is , there exists two integers $m<n$ and a configuration $\{v_i\}_{i\in \mathbb{Z}}$ satisfying
$v_i = u_i^{\beta_0}$ for any $i\leq m $ or $i> n$ such that
\begin{equation}\label{contradiction}
\mathscr{S}_m^n(\{v_i\}_{\i\in\mathbb{Z}})\equiv\sum_{i=m}^n  \frac{1}{2}(v_i-v_{i+1})^2-V(v_i\alpha) < \sum_{i=m}^n  \frac{1}{2}(u_i^{\beta_0}-u_{i+1}^{\beta_0})^2-V(u_i^{\beta_0}\alpha).
\end{equation}

Since $\min_{v\in\mathbb{R}^{\mathbb{Z}}} \mathscr{S}_m^n(v)$ is a
minimizing problem of finite variables and $\mathscr{S}_m^n(v)$ is
bounded from below, there exists a minimizing segment
$\{w_i\}_{i=m}^{n+1}$ of $\mathscr{S}_m^n$ with the boundary condition
$w_i=u_i^{\beta_0}$ for $i = m $ or $i= n+1$. One can suppose, without
loss of generality, that there exists $m<i_0\leq n$ such that $w_{i_0}
> u_{i_0}^{\beta_0}$. Since $u^{\beta}$ is a foliation, there exist
$\beta_1>\beta_0$ and $m<i_1\leq n$ such that
\begin{equation}\label{contradiction}
w_{i_1} = u_{i_1}^{\beta_1}, \text{ and }w_i\leq u_i^{\beta_1},\quad \forall~m\leq i\leq n+1.
\end{equation}

Indeed, one can choose $m<i_2\leq n$ such that $w_{i_2} = u_{i_2}^{\beta_1}$ and $w_{i_2-1} < u_{i_2-1}^{\beta_1}$.

We use and adaptation of 
the  standard technique of the Hilbert integral in calculus of variations
(see also \cite{CLlave'98}.  For every $m\leq i \leq n$ we calculate
\begin{equation}
\begin{split}
0 &= \partial_x H(w_i, w_{i+1}) + \partial_y H(w_{i-1}, w_i) + \partial_x H(u^{\beta}_i, u^{\beta}_{i+1}) + \partial_y H(u^{\beta}_{i-1}, u^{\beta}_i ) \\
&= \int_0^1 \frac{d}{dt}  \left[ \partial_x H(t w_i + (1-t)u^{\beta}_i, t w_{i+1} +(1-t)u^{\beta}_{i+1})\right. \\
&\qquad\qquad+\left. \partial_y H( t w_{i-1} + (1-t) u^{\beta}_{i-1}, t w_i +(1-t) u^{\beta}_i)\right] \ dt \\
&= \int_0^1 \left[(\partial_{xx} H) (w_i - u_i^{\beta} ) + (\partial_{xy} H) (w_{i+1} - u_{i+1}^{\beta})\right.\\
&\qquad\qquad +\left. (\partial_{yx} H) (w_{i-1} - u_{i-1}^{\beta} ) + (\partial_{yy} H) (w_{i} - u_{i}^{\beta})\right] \ dt.
\end{split}
\end{equation}

Let $i=i_2,~ \beta=\beta_1$ in the above calculation, we obtain
\[
0= w_{i_2+1} - u_{i_2+1}^{\beta_1} + w_{i_2-1} - u_{i_2-1}^{\beta_1}.
\]
Hence, due to the choice of $i_2$, we have $ w_{i_2+1} - u_{i_2+1}^{\beta_1} = u_{i_2-1}^{\beta_1} - w_{i_2-1} >0$, 
 which contradicts \eqref{contradiction}.
\end{proof}

\section*{Acknowledgements} 
The authors thank  the hospitality of the 
JLU-GT joint institute for theoretical Sciences, which 
made possible the collaboration.  R. L. also thanks
Beijing Normal University for hospitality.

\bibliographystyle{alpha}
\bibliography{QPreference}

\begin{thebibliography}{KdlLR97}

\bibitem[ALD83]{ALD}
S.~Aubry and P.~Y. Le~Daeron.
\newblock The discrete {F}renkel-{K}ontorova model and its extensions. {I}.
  {E}xact results for the ground-states.
\newblock {\em Phys. D}, 8(3):381--422, 1983.

\bibitem[BdlL13]{BL'13}
Timothy Blass and Rafael de~la Llave.
\newblock The analyticity breakdown for {F}renkel-{K}ontorova models in
  quasi-periodic media: numerical explorations.
\newblock {\em J. Stat. Phys.}, 150(6):1183--1200, 2013.

\bibitem[BK04]{BraunK}
O.~M. Braun and Y.~S. Kivshar.
\newblock {\em The {F}renkel-{K}ontorova model}.
\newblock Texts and Monographs in Physics. Springer-Verlag, Berlin, 2004.

\bibitem[Bla89]{Blank'89}
M.~L. Blank.
\newblock Metric properties of minimal solutions of discrete periodical
  variational problems.
\newblock {\em Nonlinearity}, 2(1):1--22, 1989.

\bibitem[Bla90]{Blank'90}
M.~L. Blank.
\newblock Chaos and order in the multidimensional {F}renkel-{K}ontorova model.
\newblock {\em Teoret. Mat. Fiz.}, 85(3):349--367, 1990.

\bibitem[Car99]{Caratheodory}
C.~Carath{\'e}odory.
\newblock {\em Calculus of variations and partial differential equations of the
  first order}.
\newblock Translated from the German by Robert B. Dean, Julius J. Brandstatter,
  translating editor. AMS Chelsea Publishing, 1999.

\bibitem[CdlL98]{CLlave'98}
A.~Candel and R.~de~la Llave.
\newblock On the {A}ubry-{M}ather theory in statistical mechanics.
\newblock {\em Comm. Math. Phys.}, 192(3):649--669, 1998.

\bibitem[CdlL01]{CaffarelliL}
Luis~A. Caffarelli and Rafael de~la Llave.
\newblock Planelike minimizers in periodic media.
\newblock {\em Comm. Pure Appl. Math.}, 54(12):1403--1441, 2001.

\bibitem[D{\'a}v13]{Davila}
Gonzalo D{\'a}vila.
\newblock Plane-like minimizers for an area-{D}irichlet integral.
\newblock {\em Arch. Ration. Mech. Anal.}, 207(3):753--774, 2013.

\bibitem[dlL08]{Rafael'08}
Rafael de~la Llave.
\newblock K{AM} theory for equilibrium states in 1-{D} statistical mechanics
  models.
\newblock {\em Ann. Henri Poincar\'e}, 9(5):835--880, 2008.

\bibitem[dlLV07a]{LV'07b}
Rafael de~la Llave and Enrico Valdinoci.
\newblock Critical points inside the gaps of ground state laminations for some
  models in statistical mechanics.
\newblock {\em J. Stat. Phys.}, 129(1):81--119, 2007.

\bibitem[dlLV07b]{LV'07}
Rafael de~la Llave and Enrico Valdinoci.
\newblock Ground states and critical points for generalized
  {F}renkel-{K}ontorova models in {$\Bbb Z^d$}.
\newblock {\em Nonlinearity}, 20(10):2409--2424, 2007.

\bibitem[dlLV09]{ValdinociL}
Rafael de~la Llave and Enrico Valdinoci.
\newblock A generalization of {A}ubry-{M}ather theory to partial differential
  equations and pseudo-differential equations.
\newblock {\em Ann. Inst. H. Poincar\'e Anal. Non Lin\'eaire},
  26(4):1309--1344, 2009.

\bibitem[dlLV10]{LV'10}
Rafael de~la Llave and Enrico Valdinoci.
\newblock Ground states and critical points for {A}ubry-{M}ather theory in
  statistical mechanics.
\newblock {\em J. Nonlinear Sci.}, 20(2):153--218, 2010.

\bibitem[FK39]{FK}
J.~Frenkel and T.~Kontorova.
\newblock On the theory of plastic deformation and twinning.
\newblock {\em Acad. Sci. URSS, J. Physics}, 1:137--149, 1939.

\bibitem[KdlLR97]{KLR'97}
Hans Koch, Rafael de~la Llave, and Charles Radin.
\newblock Aubry-{M}ather theory for functions on lattices.
\newblock {\em Discrete Contin. Dynam. Systems}, 3(1):135--151, 1997.

\bibitem[LS03]{LS'03}
Pierre-Louis Lions and Panagiotis~E. Souganidis.
\newblock Correctors for the homogenization of {H}amilton-{J}acobi equations in
  the stationary ergodic setting.
\newblock {\em Comm. Pure Appl. Math.}, 56(10):1501--1524, 2003.

\bibitem[Ma{\~n}97]{Mane'97}
Ricardo Ma{\~n}{\'e}.
\newblock Lagrangian flows: the dynamics of globally minimizing orbits.
\newblock {\em Bol. Soc. Brasil. Mat. (N.S.)}, 28(2):141--153, 1997.

\bibitem[Mat82]{Mather'82}
John~N. Mather.
\newblock Existence of quasiperiodic orbits for twist homeomorphisms of the
  annulus.
\newblock {\em Topology}, 21(4):457--467, 1982.

\bibitem[Mat91]{Mather'91}
John~N. Mather.
\newblock Action minimizing invariant measures for positive definite
  {L}agrangian systems.
\newblock {\em Math. Z.}, 207(2):169--207, 1991.

\bibitem[Mor24]{Morse'24}
Harold~Marston Morse.
\newblock A fundamental class of geodesics on any closed surface of genus
  greater than one.
\newblock {\em Trans. Amer. Math. Soc.}, 26(1):25--60, 1924.

\bibitem[Mos86]{Moser86}
J{\"u}rgen Moser.
\newblock Minimal solutions of variational problems on a torus.
\newblock {\em Ann. Inst. H. Poincar\'e Anal. Non Lin\'eaire}, 3(3):229--272,
  1986.

\bibitem[Mos89]{Moser89}
J{\"u}rgen Moser.
\newblock Minimal foliations on a torus.
\newblock In {\em Topics in calculus of variations ({M}ontecatini {T}erme,
  1987)}, volume 1365 of {\em Lecture Notes in Math.}, pages 62--99. Springer,
  Berlin, 1989.

\bibitem[RS11]{RabinowitzS}
Paul~H. Rabinowitz and Edward~W. Stredulinsky.
\newblock {\em Extensions of {M}oser-{B}angert theory}.
\newblock Progress in Nonlinear Differential Equations and their Applications,
  81. Birkh\"auser/Springer, New York, 2011.
\newblock Locally minimal solutions.

\bibitem[Rue69]{Ruelle}
David Ruelle.
\newblock {\em Statistical mechanics: {R}igorous results}.
\newblock W. A. Benjamin, Inc., New York-Amsterdam, 1969.

\bibitem[SdlL12a]{SuL2}
Xifeng Su and Rafael de~la Llave.
\newblock K{AM} theory for quasi-periodic equilibria in 1{D} quasi-periodic
  media: {II}. {L}ong-range interactions.
\newblock {\em J. Phys. A}, 45(45):455203, 24, 2012.

\bibitem[SdlL12b]{SuL1}
Xifeng Su and Rafael de~la Llave.
\newblock K{AM} {T}heory for {Q}uasi-periodic {E}quilibria in
  {O}ne-{D}imensional {Q}uasi-periodic {M}edia.
\newblock {\em SIAM J. Math. Anal.}, 44(6):3901--3927, 2012.

\bibitem[Sel92]{Selke2}
W~Selke.
\newblock Spatially modulated structures in systems with competing
  interactions.
\newblock In {\em Phase transitions and critical phenomena, Volume 15}, pages
  1--72. Academic Press, 1992.

\bibitem[SZdlL15]{SuZL15}
Xifeng Su, Lei Zhang, and Rafael de~la Llave.
\newblock Resonant equilibrium configurations in quasi-periodic media i:
  perturbative expansions.
\newblock 2015.
\newblock Preprint available at http://arxiv.org/abs/1503.03304.

\bibitem[Tor04]{Torres}
Monica Torres.
\newblock Plane-like minimal surfaces in periodic media with exclusions.
\newblock {\em SIAM J. Math. Anal.}, 36(2):523--551, 2004.

\bibitem[Val04]{Valdinoci}
Enrico Valdinoci.
\newblock Plane-like minimizers in periodic media: jet flows and
  {G}inzburg-{L}andau-type functionals.
\newblock {\em J. Reine Angew. Math.}, 574:147--185, 2004.

\bibitem[ZSdlL15]{ZhangSL15}
Lei Zhang, Xifeng Su, and Rafael de~la Llave.
\newblock Equlibrium quasi-periodic configurations in quasi-periodic media with
  resonant frequencies {II}: {KAM} theory.
\newblock 2015.
\newblock Preprint available at http://arxiv.org/abs/1503.03311.

\end{thebibliography}
\end{document}